\documentclass{amsart}

\usepackage[letterpaper,margin=1in]{geometry}
\usepackage{amsmath,amssymb,amsthm,mathtools}
\usepackage{booktabs,array}
\usepackage{xcolor}
\usepackage{microtype}
\usepackage[colorlinks=true,linkcolor=blue!55!black,citecolor=blue!55!black,urlcolor=blue!55!black]{hyperref}
\usepackage{enumerate}

\numberwithin{equation}{section}

\newtheorem{theorem}{Theorem}[section]
\newtheorem{proposition}[theorem]{Proposition}
\newtheorem{lemma}[theorem]{Lemma}

\newtheorem{definition}[theorem]{Definition}
\newtheorem{conjecture}[theorem]{Conjecture}
\theoremstyle{remark}

\newcommand{\tr}{\tau}

\newcommand{\E}{\mathbb E}
\newcommand{\one}{\mathbf 1}
\newcommand{\ip}[2]{\langle #1,#2\rangle}
\newcommand{\norm}[1]{\lVert #1\rVert}
\newcommand{\comm}[2]{[#1,#2]}
\newcommand{\ind}{\mathbf 1}

\newcommand{\dc}{\mathbf{d}}
\newcommand{\var}{{\mathrm{Var}}}
\newcommand{\Rea}{\operatorname{Re}}
\newcommand{\id}{\operatorname{id}}

\title{The Quantum KKL Inequality}

\author{Yong Jiao}
\address{School of Mathematics and Statistics, HNP-LAMA, Central South University, Changsha 410075, China}
\email{jiaoyong@csu.edu.cn}

\author{Wenlong Lin}
\address{School of Mathematics and Statistics, HNP-LAMA, Central South University, Changsha 410075, China}
\email{linwenlong2000@foxmail.com}

\author{Sijie Luo}
\address{School of Mathematics and Statistics, HNP-LAMA, Central South University, Changsha 410075, China}
\email{sijieluo@csu.edu.cn}

\author{Dejian Zhou}
\address{School of Mathematics and Statistics, HNP-LAMA, Central South University, Changsha 410075, China}
\email{zhoudejian@csu.edu.cn}

\date{}

\begin{document}
\raggedbottom

\begin{abstract}
In this paper, we resolve the quantum KKL conjecture of Montanaro and Osborne \cite{MO2010} for quantum Boolean functions on the $n$-qubit hypercube.  More precisely, for every self-adjoint unitary $T$, we bound the largest $L_2$-influence from below by a constant multiple of $\mathrm{Var}(T)\log(n)/n$.  The proof heavily depends on upper and lower commutator estimates on the Hilbert space.
\end{abstract}
\subjclass[2020]{Primary 46L53; Secondary 47D07, 94D10.}

\keywords{Quantum KKL inequality, quantum Boolean functions, noncommutative hypercontractivity, heat semigroup}

\maketitle

\section{Introduction}

 For fixed $n\in\mathbb{N}$, let $\{-1,1\}^n$ be the hypercube equipped with the uniform probability measure $\mu_n$, and let $L_{p}(\{-1,1\}^{n})$ be the associated $L_{p}$ space for $1\leq p\leq \infty$. In the sequel, we will use the shorthand notation $[n]\coloneqq\{1,2,\cdots,n\}$. For each $j\in[n]$, the $j$-th influence of $f:\{-1,1\}^n\to\mathbb{R}$ is given by
$$\mathrm{Inf}_j(f)=\mu_n(\{x\in \{-1,1\}^n|\, f(x)\neq f(x^{\oplus j})\}),$$
where $x^{\oplus j}$ means flipping the $j$-th variable of $x$, i.e. for $x=(x_1,\cdots, x_n)\in \{-1,1\}^n$, 
$$x^{\oplus j}=(x_1,\cdots,x_{j-1},-x_j,x_{j+1},\cdots,x_n).$$
The total influence of $f$ is defined by $\mathrm{Inf}(f)=\sum_{j\in[n]}\mathrm{Inf}_j(f)$,  which is often used to measure the complexity of the function $f$. The $j$-th partial derivative of the $f:\{-1,1\}^{n}\to \mathbb{R}$ is defined as follows:
\begin{equation*}
\dc_{j}(f)(x)\coloneqq\frac{f(x)-f(x^{\oplus j})}{2},\quad x\in\{-1,1\}^{n}.
\end{equation*}
For each Boolean function $f:\{-1,1\}^n\to\{-1,1\}$ and $j\in [n]$, it is easy to verify that
\begin{equation}\label{Inf-Dj}
\mathrm{Inf}_j(f)=\|\dc_{j}(f)\|_{L_2(\{-1,1\}^n)}^2=\|\dc_{j}(f)\|_{L_p(\{-1,1\}^n)}^p,\quad \forall 1\leq p< \infty.
\end{equation}
A Boolean function $f:\{-1,1\}^n\to\{-1,1\}$ is said to be balanced whenever $\mathbb{E}_{\mu_n}(f)=0$, that is, $\mathrm{Var}(f)=1$. 
In the remarkable paper \cite{KKL1988}, Kahn,  Kalai and Linial proved that for each balanced Boolean
function, there exists some universal constant $C>0$ such that 
\begin{equation}\label{KKL}
\max_{j\in[n]}\mathrm{Inf}_j(f)\geq C\frac{\log(n)}{n}. 
\end{equation}
Thus, some variable has an influence at least $O(\log(n)/n)$ which is larger than the order $1/n$ deduced from the Poincar\'e inequality.
Since then, this inequality of Kahn,  Kalai and Linial (KKL inequality in short) plays an important role in Boolean analysis and has been generalized and strengthened by many mathematicians in different directions (see e.g. \cite{EG2022, Fr1998, Ta1994}).

 One of significant improvements of the KKL inequality is the Talagrand ($L_{1}$-$L_{2}$-) influence inequality established in  \cite{Ta1994}. More precisely, Talagrand proved that: 
There exists a universal constant $C>0$ such that for each function $f:\{-1,1\}^{n}\to\mathbb{R}$ the following holds
\begin{equation}\label{Ta-L1L2}
\var(f)\leq C\sum_{j=1}^{n}\frac{\|\dc_{j}(f)\|^{2}_{L_{2}(\{-1,1\}^{n})}}{1+\log(\|\dc_{j}(f)\|_{L_{2}(\{-1,1\}^{n})}/\|\dc_{j}(f)\|_{L_{1}(\{-1,1\}^{n})})}.
\end{equation}
It is easy to check that for each balanced Boolean function  \eqref{Ta-L1L2} implies the KKL inequality \eqref{KKL} via \eqref{Inf-Dj}.

Another improvement of the KKL inequality was recently obtained by Eldan and Gorss \cite{EG2020} (see also \cite{EG2022}).
Particularly, motivated by a conjecture of Talagrand \cite{Ta1997}, Eldan and Gorss \cite{EG2020} (see also \cite{EG2022}) applied stochastic analysis techniques to prove that
there exists a universal constant $C>0$ such that 
\begin{equation}\label{EG-2022}
\var(f)\sqrt{\log\left(1+\frac{e}{\sum_{j=1}^{n}\|\dc_{j}(f)\|^{2}_{L_{1}(\{-1,1\}^{n})}}\right)}\leq C\left\|\left(\sum_{j=1}^{n}|\dc_{j}(f)|^{2}\right)^{1/2}\right\|_{L_{1}(\{-1,1\}^{n})},
\end{equation}
for every Boolean function $f:\{-1,1\}^{n}\to\{-1,1\}$. One may find the detailed proof that how the Eldan-Gorss inequality \eqref{EG-2022} deduces the KKL inequality \eqref{KKL}.

 We also refer interested readers to \cite{BKKKL1992,CL2012,EG2022,KKKMS2021,OW2013} for further extensions and applications of  the KKL inequality, the Talagrand influence inequality, and to \cite{BIM2023,EKLM2022,IZ2024,Ros2020} for new proofs of the  Eldan-Gross inequality \eqref{EG-2022}.

This paper is mainly concerned with the quantum version of the KKL inequality. To this end, we first introduce the basic notation.  A quantum analogue of functions on the Boolean hypercubes, i.e., functions of $n$ bits, is the observables on $n$ qubits. Then the quantum hypercube is the  algebra of observables $\mathcal{M}_{2^n}:=M_2(\mathbb{C})^{\otimes n}$ equipped with the normalized trace $\tau:=\frac{1}{2^n}\mathrm{Tr}$ (here $\mathrm{Tr}$ is the usual trace on matrix). For $1\leq p<\infty$, define the Lebesgue norm by setting
\[\|T\|_{p}=\Big(\tau(|T|^p)\Big)^{1/p},\quad T\in \mathcal{M}_{2^n}.\]
The following definition of quantum Boolean function is taken from \cite[Definition 3.1]{MO2010}.
\begin{definition}
	An operator $T\in\mathcal{M}_{2^n}$ is called a \emph{quantum Boolean function}
	if it is self-adjoint and unitary, that is,
	\[
	T=T^*,\qquad T^2=\one.
	\]
\end{definition}
For each $j\in [n]$, let $D_{j}$ be the $j$-th partial derivative operator on $\mathcal{M}_{2^n}$ given by 
\begin{equation}\label{eq:Dj}
D_j(T)=\left(\mathbb I^{\otimes(j-1)}\otimes\mathcal P
\otimes\mathbb I^{\otimes(n-j)}\right)(T),
\quad T\in\mathcal{M}_{2^n}.
\end{equation}
Here $\mathbb{I}$ denotes the identity map on $M_2(\mathbb{C})$, and $\mathcal{P}:M_2(\mathbb{C})\to M_2(\mathbb{C})$ is a projection given by 
\[\mathcal{P}(A)=A- \frac{1}{2}\mathrm{Tr}(A) \one_2.\]
Then the quantum $L_p$-influence of given $T\in \mathcal{M}_{2^n}$ in the $j$-th coordinate is defined as
\[\mathrm{Inf}_j^p(T):=\|D_jT\|_p^p,\quad 1\leq p<\infty.\]
Particularly, when $p=1$, the $L_1$-influence  $\mathrm{Inf}_j^1(T)=\|D_j(T)\|_1$ is also called the geometric influence (for its relation to isoperimetric inequalities in the classical setting). 

 In \cite[Proposition 11.1]{MO2010}, Montanaro and Osborne derived a quantum analogy of the Talagrand influence inequality \eqref{Ta-L1L2} via the quantum hypercontractivity principle. However, we cannot obtain a quantum KKL inequality as in the classical hypercube setting since now the quantum version of \eqref{Inf-Dj} fails. More precisely, the identity 
$\|D_jT\|_2^2=\|D_jT\|_1$  may \emph{fail} for some quantum Boolean function. Actually, in \cite{MO2010}, Montanaro and Osborne provided some concrete example such that   $\|D_jT\|_1=\|D_jT\|_2$.  Nevertheless, they proved that the quantum KKL inequality holds for some operators \(T\) satisfying $\|D_jT\|_1=\|D_jT\|_2$. This led them to formulate the quantum KKL conjecture below. Denote by
$$\var(T)=\tau(|T|^2)-\tau(T)^2$$
  the variance of $T\in \mathcal{M}_{2^n}$. In particular, if $T$ is a quantum Boolean function, then $
 \mathrm{Var}(T) =1-\tau(T)^2$. 
 
 \begin{conjecture}[Quantum KKL conjecture]\label{QKKL}
There exists a universal constant $C>0$ such that for each $n\in \mathbb{N}$ and quantum Boolean function $T\in \mathcal{M}_{2^n}$ the following holds
\[
\max_{j\in [n]}\|D_{j}(T)\|^{2}_{L_{2}(\mathcal{M}_{2^n})}\geq C\var(T) \frac{\log(n)}{n}.
\]
\end{conjecture}

Subsequently, numerous researchers, including the present authors, attempted to resolve the quantum KKL conjecture. The first significant progress was made by Rouz\'{e},  Wirth, and Zhang \cite{RWZ2024} in the sense of $L_p$-influence with $1\leq p<2$. In fact, they proved that (\cite[Theorem 3.9]{RWZ2024})
\[\max_{j\in [n]}\|D_jT\|_p^p\geq C_p \var(T) \frac{\log(n)}{n},\quad 1\leq p<2,\]
for each quantum Boolean function $T\in\mathcal{M}_{2^n}$. Unfortunately, the constant \(C_p\) blows up as \(p \to 2\). Consequently, the KKL conjecture remains unresolved by their result.  Recently, the authors \cite{JLZZ2024} obtained the same result of  Rouz\'{e},  Wirth, and Zhang via quantum restriction method  and quantum Eldan-Gross inequality (see \cite{BGX2024} for a independent proof). Our previous work \cite{JLZ2025} shows that the noncommutative $L_2$-influence KKL inequality fails in the CAR algebra setting. In this paper, we continue this line of research and establish a quantum KKL inequality for \(L_2\)-influences, thereby resolving the KKL conjecture affirmatively.

For brevity, we introduce the total influence of $T\in \mathcal{M}_{n}$ by
\begin{equation*}
\mathrm{Inf}(T):=\sum_{j\in[n]}\mathrm{Inf}_j^2(T)=\sum_{j\in[n]}\norm{D_jT}_2^2.
\end{equation*}
A quantum Boolean function $T$ is said to be nontrivial if $T\neq \mathbf{1}$. Our main result of this paper is read as follows.

\begin{theorem}\label{thm:main}
	Assume that $T\in \mathcal{M}_{2^n}$ is a nontrivial quantum Boolean function.  Then
	\begin{equation}\label{eq:exp-form}
		\log\frac{1}{\max_{j\in[n]}\norm{D_jT}_2^2}
		\le 2^{11} \left(\frac{\mathrm{Inf}(T)}{\mathrm{Var}(T)}\right).
	\end{equation}
As an consequence, the following quantum KKL inequality holds
\[
\max_{j\in [n]}\norm{D_jT}_2^2\geq \frac{\mathrm{Var}(T)}{2^{12}}\left(\frac{\log n}{n}\right).
	\]
\end{theorem}

We now outline novelties of our approach. The main feature of our argument is that we do not attempt to derive the $L_{2}$ quantum KKL inequality from a global functional inequality of Talagrand, logarithmic Sobolev, or entropy type. Such inequalities have proved extremely effective for several quantum analogues of classical influence inequalities, especially for geometric or $L_{p}$ influence (see e.g. \cite{JLZ2025,JLZZ2024,RWZ2024}). However, due the failure of $\|D_{j}T\|^{2}_{2}=\|D_{j}T\|_{1}$ for quantum Boolean function $T$, the global functional inequality approach does not capture enough information about the noncommutativity. Our starting point is therefore different.  Rather than estimating the derivatives $D_{j}$ via heat semigroup tools, we decompose them through commutator derivations. For
$j\in \{1,\cdots, n\}$ and $a\in \{1,2,3\}$, let $\widetilde{\operatorname{ad}}_{ja}(T)\coloneqq \frac{i}{\sqrt 8}[\sigma_{j}^{a},T]$. The normalization is chosen so that
\[
\sum_{a=1}^{3}\widetilde{\operatorname{ad}}_{ja}^{*}\widetilde{\operatorname{ad}}_{ja}=D_{j}, \qquad \sum_{j=1}^{n}\sum_{a=1}^{3}\widetilde{\operatorname{ad}}_{ja}^{*}\widetilde{\operatorname{ad}}_{ja}=\Delta.
\]
Consequently, $\sum_{a=1}^{3}\|\delta_{ja}T\|_{2}^{2}=\|D_{j}T\|_{2}^{2}$ and $\sum_{j,a}\|\delta_{ja}T\|_{2}^{2}=\mathrm{Inf}(T)$. Thus the local derivatives $\{\widetilde{\operatorname{ad}}_{ja}\}_{j, a}$ do not merely provide an auxiliary differential structure. They provide an exact factorization of the influence operators. This is the reason that commutators are naturally adapted to the quantum KKL conjecture.

There is a second reason why the commutator derivations are particularly well suited to quantum Boolean functions.  If $T=T^{*}$ and $T^{2}=\mathbf
1$, then the Leibniz rule gives
\[
 T\widetilde{\operatorname{ad}}_{ja}(T)+\widetilde{\operatorname{ad}}_{ja}(T)T=0.
\]
Hence every local derivative $\widetilde{\operatorname{ad}}_{ja}(T)$ fulfills certain anti-commuting relation of $T$. This is the point at which noncommutativity becomes an advantage rather than an obstruction.

The rest of the paper is organized as follows.  Section \ref{sec 2} develops the Pauli expansion, the normalized commutator derivations, and the gradient estimates used later.  Section \ref{sec 3} constructs a commuting Pauli dilation and proves the asymmetric commutator comparison.  Section \ref{sec 4} builds heat-semigroup test operators and establishes the corresponding lower commutator estimate.  Finally, Section \ref{sec 5} combines the upper and lower commutator bounds to prove our main result Theorem~\ref{thm:main}.

\section{Preliminaries}\label{sec 2}

\subsection{Boolean analysis on Quantum hypercube}
The Pauli matrices in $M_2(\mathbb{C})$ are as follows
\[
  \sigma^0=\one_2=\begin{pmatrix}1&0\\0&1\end{pmatrix},\qquad
  \sigma^1=\begin{pmatrix}0&1\\1&0\end{pmatrix},\qquad
  \sigma^2=\begin{pmatrix}0&-i\\i&0\end{pmatrix},\qquad
  \sigma^3=\begin{pmatrix}1&0\\0&-1\end{pmatrix}.
\]
They satisfy
\begin{equation*}
  (\sigma^a)^*=\sigma^a,\qquad
  (\sigma^a)^2=\one_2,\qquad
  \tr_2(\sigma^a\sigma^b)=\delta_{ab}
  \quad(0\le a,b\le3),
\end{equation*}
and, for $1\le a,b\le3$,
\[
  \sigma^a\sigma^b
  =\delta_{ab}\one_2+i\sum_{c=1}^3\epsilon_{abc}\sigma^c.
\]
Here $\delta_{ab}$ is the Kronecker symbol, and $\epsilon_{abc}$ is the
symbol given by:
\begin{equation*}
  \epsilon_{abc}
  =\begin{cases}
    1,&(a,b,c)\in\{(1,2,3),(2,3,1),(3,1,2)\},\\
   -1,&(a,b,c)\in\{(1,3,2),(3,2,1),(2,1,3)\},\\
    0,&\text{two of }a,b,c\text{ are equal}.
  \end{cases}
\end{equation*}
Moreover, $(\sigma^a)_{a\in \{0,1,2,3\}}$  forms an orthonormal basis of $M_2(\mathbb{C})$. For $\mathbf{s}\in\{0,1,2,3\}^n$, set
\[
  \sigma_{\mathbf{s}}=\bigotimes_{j=1}^n\sigma^{s_j}.
\]
Then $(\sigma_{\mathbf{s}})_{\mathbf{s}}$ forms an orthonormal basis of $\mathcal{M}_{2^n}=M_2(\mathbb{C})^{\otimes n}$.
Accordingly, each $T\in \mathcal{M}_{2^n}$ can uniquely be expressed as
\[
T=\sum_{\mathbf{s}\in\{0,1,2,3\}^n}\widehat{T}(\mathbf{s})\sigma_{\mathbf{s}},
\qquad
\widehat{T}(\mathbf{s})=\tr(\sigma_{\mathbf{s}}T).
\]
Here $\widehat{T}(\mathbf{s})$ is the Fourier coefficient. Note that if $T\in \mathcal{M}_{2^n}$ is self-adjoint, then $(\widehat{T}(\mathbf{s}))_{\mathbf{s}} $ are all real numbers. 

\noindent Let
\[
\Delta=\sum_{j=1}^nD_j,
\qquad P_t=e^{-t\Delta}.
\]
Given $\mathbf{s}\in\{0,1,2,3\}^n$, denote
\[\qquad
\operatorname{supp}(\mathbf{s})=\{j\in[n]:s_j\ne0\},
\qquad
|\mathbf{s}|=\#\operatorname{supp}(\mathbf{s}),\]
where $\# B$ denotes the cardinality of a finite set $ B$. 
According to the definition of $D_j$ as in \eqref{eq:Dj}
\begin{equation}\label{eq:Ej-on-Pauli}
  D_j(\sigma_{\mathbf{s}})
  =\ind_{\{s_j\ne0\}}\sigma_{\mathbf{s}}.
\end{equation}
Then $\Delta\sigma_{\mathbf{s}}=|\mathbf{s}|\sigma_{\mathbf{s}}$ and
$P_t\sigma_{\mathbf{s}}=e^{-t|\mathbf{s}|}\sigma_{\mathbf{s}}$ for each $\mathbf{s}\in\{0,1,2,3\}^n$.
For each  $T\in \mathcal{M}_{2^n}$, 
Parseval's identity and \eqref{eq:Ej-on-Pauli} now give
\begin{align}
  \norm T_2^2
    &=\sum_{\mathbf{s}}|\widehat{T}(\mathbf{s})|^2,
      \notag\\
  \norm{D_jT}_2^2
    &=\sum_{\mathbf{s}:s_j\ne0}|\widehat{T}(\mathbf{s})|^2,
      \notag\\
  \mathrm{Inf}(T):=\sum_{j\in [n]}\norm{D_jT}_2^2
    &=\ip{T}{LT}=\sum_{\mathbf{s}}|\mathbf{s}|\,
      |\widehat{T}(\mathbf{s})|^2.
      \label{eq:total-inf}
\end{align}

  If
$X\ne0$, define its Pauli degree by
\[
  \deg X:=\max\{|\mathbf{s}|:\widehat X(\mathbf{s})\ne0\},
  \qquad \deg0:=0.
\]
The projection onto Pauli degrees at most $m$ is defined by setting
\[
  \Pi_{\le m}X
  :=\sum_{|\mathbf{s}|\le m}
    \widehat X(\mathbf{s})\sigma_{\mathbf{s}}.
\]
Then
\begin{equation}\label{eq:parent-tail}
  \norm{T-\Pi_{\le m}T}_2^2
  =\sum_{|\mathbf{s}|>m}\widehat{T}(\mathbf{s})^2\leq\frac1{m+1}
  \sum_{|\mathbf{s}|>m}|\mathbf{s}|\widehat{T}(\mathbf{s})^2
  \le \frac{\mathrm{Inf}(T)}{m+1}.
\end{equation}

\subsection{Commutator estimates}
Throughout this manuscript, square brackets denote the
commutator:
\begin{equation*}
  [A,B]:=AB-BA.
\end{equation*}
\noindent For $S,X\in\mathcal{M}_{2^n}$, define
\begin{equation*}
  \operatorname{ad}_S(X):=[S,X]=SX-XS.
\end{equation*}
We later also use the normalized derivation
\begin{equation}\label{eq:adtilde-def}
  \widetilde{\operatorname{ad}}_S(X)
  :=\frac{i}{\sqrt8}\operatorname{ad}_S(X).
\end{equation}
For $1\le j\le n$, let
\[
e_j=(0,\ldots,0,\underset{j\text{th position}}{1},0,\ldots,0)
\in\{0,1\}^n
\]
be the $j$-th standard basis vector.  For $1\le a\le3$, the multi-index
$a e_j$ belongs to $\{0,1,2,3\}^n$, and
single-qubit Pauli operator by
\begin{equation*}
	\sigma_{a e_j}
	:=\one_2^{\otimes(j-1)}\otimes\sigma^a
	\otimes\one_2^{\otimes(n-j)}.
\end{equation*}
For $j\in[n]$ and $a\in[3]$, we write for convenience
\[
  \operatorname{ad}_{ja}:=  \operatorname{ad}_{\sigma_{a e_j}},\qquad \widetilde{\operatorname{ad}}_{ja}
  :=\widetilde{\operatorname{ad}}_{\sigma_{a e_j}}.
\]
The following lemma is the Leibniz rule for the mapping $\operatorname{ad}_{S}$ for $S\in \mathcal{M}_{2^{n}}$.
\begin{lemma}\label{lem:adS}
For $S,X,Y\in\mathcal{M}_{2^n}$:
\begin{enumerate}[{\rm (i)}]
  \item $\operatorname{ad}_S^*=\operatorname{ad}_{S^*}$; in
  particular, $\operatorname{ad}_S$ is self-adjoint on
  $L_2(\mathcal{M}_{2^n},\tr)$ when $S=S^*$.
  \item
  $\operatorname{ad}_S(XY)
   =\operatorname{ad}_S(X)Y+X\operatorname{ad}_S(Y)$.
  \item If $X^2=\one$, then
  $\operatorname{ad}_S(X)X+X\operatorname{ad}_S(X)=0$.
  \item If $S^2=\one$, then
  $\operatorname{ad}_S^2(X)=2(X-SXS)$.
\end{enumerate}
\end{lemma}

\begin{proof}
(i) For $A,B\in\mathcal{M}_{2^n}$,
\begin{align*}
  \ip{A}{\operatorname{ad}_S(B)}
  &=\tr(A^*SB-A^*BS)\\
  &=\tr\bigl((A^*S-SA^*)B\bigr)
   =\ip{\operatorname{ad}_{S^*}(A)}{B}.
\end{align*}

(ii) We have 
\begin{align*}
  \operatorname{ad}_S(XY)
  &=SXY-XYS\\
  &=(SX-XS)Y+X(SY-YS)\\
  &=\operatorname{ad}_S(X)Y+X\operatorname{ad}_S(Y).
\end{align*}

(iii) If $X^2=\one$, then
\[
  0=\operatorname{ad}_S(\one)
   =\operatorname{ad}_S(X^2)
   =\operatorname{ad}_S(X)X+X\operatorname{ad}_S(X).
\]

(iv) If $S^2=\one$, then
\[
  \operatorname{ad}_S^2(X)
  =S(SX-XS)-(SX-XS)S
  =2(X-SXS).
\]
\end{proof}

\begin{lemma}\label{lem:adja}
For $j\in[n]$ and $a\in\{1,2,3\}$: we have the following:
\begin{enumerate}[{\rm (i)}]
  \item If $X\in \mathcal{M}_{2^{n}}$ and $X$ is self-adjoint, then
  $\widetilde{\operatorname{ad}}_{ja}(X)^*
   =\widetilde{\operatorname{ad}}_{ja}(X)$.
  \item  For every $X\in\mathcal{M}_{2^n}$,
  $\deg(\widetilde{\operatorname{ad}}_{ja}(X))\le\deg X$.
  \item For every $X\in\mathcal{M}_{2^n}$,
  \[
    \Delta\widetilde{\operatorname{ad}}_{ja}(X)
    =\widetilde{\operatorname{ad}}_{ja}(\Delta X).
  \]
  \item For every $j\in[n]$,
  \begin{equation*}
    \sum_{a=1}^3
    \widetilde{\operatorname{ad}}_{ja}^*
    \widetilde{\operatorname{ad}}_{ja}
    =D_j.
  \end{equation*}
  \item  For every $X\in\mathcal{M}_{2^n}$,
  \begin{equation*}
    \sum_{j\in[n]}\sum_{a=1}^3
    \norm{\widetilde{\operatorname{ad}}_{ja}(X)}_2^2
    =\mathrm{Inf}(X).
  \end{equation*}
\end{enumerate}
\end{lemma}

\begin{proof}
(i) We have 
\[
  \widetilde{\operatorname{ad}}_{ja}(X)^*
  =-\frac{i}{\sqrt8}[\sigma_{a e_j},X]^*
  =\frac{i}{\sqrt8}[\sigma_{a e_j},X^*]
  =\widetilde{\operatorname{ad}}_{ja}(X^*).
\]

(ii)
For $\mathbf{s}\in\{0,1,2,3\}^n$ and $c\in\{1,2,3\}$, let
$\mathbf{s}^{j\to c}$ be obtained from $\mathbf{s}$ by replacing $s_j$
with $c$.  Then
\begin{equation}\label{eq:adja-on-Pauli}
  \widetilde{\operatorname{ad}}_{ja}(\sigma_{\mathbf{s}})
  =
  \begin{cases}
    \displaystyle
    -\frac1{\sqrt2}\sum_{c=1}^3
       \epsilon_{a s_j c}\sigma_{\mathbf{s}^{j\to c}},
       &s_j\ne0,\\[6pt]
    0,&s_j=0.
  \end{cases}
\end{equation}
Every nonzero term in \eqref{eq:adja-on-Pauli} satisfies
\begin{equation}\label{eq:adja-support}
  \operatorname{supp}(\mathbf{s}^{j\to c})
  =\operatorname{supp}(\mathbf{s}),
  \qquad
  |\mathbf{s}^{j\to c}|=|\mathbf{s}|.
\end{equation}
Equations \eqref{eq:adja-on-Pauli}--\eqref{eq:adja-support} prove (ii) and
give
\[
  \Delta\widetilde{\operatorname{ad}}_{ja}(\sigma_{\mathbf{s}})
  =|\mathbf{s}|
    \widetilde{\operatorname{ad}}_{ja}(\sigma_{\mathbf{s}})
  =\widetilde{\operatorname{ad}}_{ja}
    (\Delta\sigma_{\mathbf{s}}),
\]
which proves (iii).

 (iv) Note that Lemma~\ref{lem:adS}(i) and Lemma \ref{lem:adS}(iv) give
\[
  \widetilde{\operatorname{ad}}_{ja}^*
  \widetilde{\operatorname{ad}}_{ja}
  =\frac18\operatorname{ad}_{ja}^2,
\]
and hence,
\begin{equation*}
  \widetilde{\operatorname{ad}}_{ja}^*
  \widetilde{\operatorname{ad}}_{ja}(\sigma_{\mathbf{s}})
  =\frac12
    \ind_{\{s_j\ne0\}}\ind_{\{a\ne s_j\}}\sigma_{\mathbf{s}}.
\end{equation*}
Therefore,
\[
  \sum_{a=1}^3
  \widetilde{\operatorname{ad}}_{ja}^*
  \widetilde{\operatorname{ad}}_{ja}(\sigma_{\mathbf{s}})
  =\ind_{\{s_j\ne0\}}\sigma_{\mathbf{s}}
  =D_j(\sigma_{\mathbf{s}}).
\]

(v) By item (iv), we have
\begin{align*}
  \sum_{j\in[n]}\sum_{a=1}^3
    \norm{\widetilde{\operatorname{ad}}_{ja}(X)}_2^2
  &=\sum_{j\in[n]}
    \ip{X}{\left(
      \sum_{a=1}^3
      \widetilde{\operatorname{ad}}_{ja}^*
      \widetilde{\operatorname{ad}}_{ja}\right)X}\\
  &=\sum_{j\in[n]}\ip{X}{D_jX}
   =\sum_{j\in[n]}\norm{D_jX}_2^2
   =\mathrm{Inf}(X). 
\end{align*}

\end{proof}

We also need the following basic commutator estimate with respect to a given function $u$. 
Precisely,
\begin{equation}\label{eq: def u}
	u(z)=\frac{z}{1+z^2}
\end{equation}
satisfies
\begin{equation*}
	|u(z)|\le\frac12,
	\qquad |u'(z)|\le1,
	\qquad u(1)=\frac12,
	\qquad u(-1)=-\frac12.
\end{equation*}
	In particular, $u(T)=T/2$ whenever $T\in\mathcal{M}_{2^n}$ is  a quantum Boolean function. The following lemma can be found in \cite{PS2011}; see e.g. Theorem 1 and (14) there. 

\begin{lemma}\label{lem:u}
	For self-adjoint $A,B\in\mathcal{M}_{2^n}$,
	\begin{equation*}
		\norm{u(A)-u(B)}_2\le\norm{A-B}_2,
		\qquad
		\norm{\comm{u(A)}B}_2\le\norm{\comm AB}_2.
	\end{equation*}
\end{lemma}

\subsection{Gradient estimate}

For self-adjoint $X\in \mathcal{M}_{2^{n}}$, define the carr\'e du champ associated with the
positive generator $\Delta$ directly by
\begin{equation}\label{eq:gamma-L}
	\Gamma(X)\coloneqq\frac{1}{2}\left((\Delta X)X+X(\Delta X)-\Delta(X^2)\right).
\end{equation}
Since $\Delta=\sum_{j=1}^nD_j$, this definition may equivalently be written as
\begin{equation}\label{eq:gamma-L-local-sum}
	\Gamma(X)
	=\frac12\sum_{j=1}^n
	\left((D_jX)X+X(D_jX)-D_j(X^2)\right).
\end{equation}
We now show that $\Gamma(X)$ is positive whenever $X\in \mathcal{M}_{2^n}$ is  self-adjoint. Indeed,  each $\widetilde{\operatorname{ad}}_{ja}(X)$ is self-adjoint if $X$ is self-adjoint, it follows that $\Gamma(X)\geq 0$ with the help of next lemma.

\begin{lemma}
	\label{lem:gamma}
	For each self-adjoint  $X\in \mathcal{M}_{2^n}$,
	\begin{equation}\label{eq:gamma-ad-square}
		\Gamma(X)
		=\sum_{j=1}^n\sum_{a=1}^3
		\bigl(\widetilde{\operatorname{ad}}_{ja}(X)\bigr)^2.
	\end{equation}
\end{lemma}

\begin{proof}
	We first derive \eqref{eq:gamma-ad-square} from the definition
	\eqref{eq:gamma-L}. 
	According to the definition of $D_j$ \eqref{eq:Dj}, we may write 
		$D_j=\mathbb I^{\otimes n}- E_j,$
	where 
	\begin{equation}\label{eq:Ej}
				E_j=\mathbb I^{\otimes(j-1)}\otimes\mathcal E
				\otimes\mathbb I^{\otimes(n-j)}, 
	\end{equation}
	with $\mathcal{E}(A)=\frac{1}{2}\mathrm{Tr}(A)\one_2$, $A\in M_2(\mathbb{C})$. 
Then,  for each $j\in [n]$, we have
	\begin{align}
		&(D_jX)X+X(D_jX)-D_j(X^2)\notag\\
		&\qquad=(X-E_jX)X+X(X-E_jX)-\bigl(X^2-E_j(X^2)\bigr)\notag\\
		&\qquad=X^2-E_j(X)X-XE_j(X)+E_j(X^2).
		\label{eq:gamma-direct-local-expansion}
	\end{align}
	
\noindent Define
	\[
	\chi_a(b)\coloneqq
	\begin{cases}
		1,&b=0\text{ or }b=a,\\
		-1,&b\in\{1,2,3\}\setminus\{a\}.
	\end{cases}
	\]
Then, for each $\mathbf{s}\in \{0,1,2,3\}^n$,
	\begin{align*}
		\sigma_{a e_j}\sigma_{\mathbf{s}}\sigma_{a e_j}
		&=\chi_a(s_j)\sigma_{\mathbf{s}},
	\end{align*}
	and 
	\begin{align*}
		\sum_{a=1}^3
		\sigma_{a e_j}\sigma_{\mathbf{s}}\sigma_{a e_j}
		&=\left(\sum_{a=1}^3\chi_a(s_j)\right)\sigma_{\mathbf{s}}
		=\bigl(4\ind_{\{s_j=0\}}-1\bigr)\sigma_{\mathbf{s}}
		=(4E_j-\mathbb{I}^{\otimes n})(\sigma_{\mathbf{s}}).
	\end{align*}
	Hence, by linearity,
	\begin{equation}\label{eq:Pauli-conjugation-operator}
		\sum_{a=1}^3\sigma_{a e_j}B\sigma_{a e_j}=4E_j(B)-B,
		\qquad B\in\mathcal{M}_{2^n}.
	\end{equation}
	Since $\sigma_{a e_j}^2=\one$, combining equations \eqref{eq:adtilde-def} and
	\eqref{eq:Pauli-conjugation-operator}, we have
	\begin{align*}
		\sum_{a=1}^3
		\bigl(\widetilde{\operatorname{ad}}_{ja}(X)\bigr)^2
		&=-\frac18\sum_{a=1}^3
		(\sigma_{a e_j}X-X\sigma_{a e_j})^2\\
		&=\frac18\sum_{a=1}^3
		\bigl(\sigma_{a e_j}X^2\sigma_{a e_j}+X^2
		-\sigma_{a e_j}X\sigma_{a e_j}X
		-X\sigma_{a e_j}X\sigma_{a e_j}\bigr)\\
		&=\frac12\bigl(X^2+E_j(X^2)-(E_jX)X-X(E_jX)\bigr).
	\end{align*}
This together with \eqref{eq:gamma-L-local-sum} and \eqref{eq:gamma-direct-local-expansion} gives the desired assertions.  
\end{proof}

It is known that
\begin{equation}\label{eq:gradient-input}
	\Gamma(P_sX)\le e^{-s}P_s\Gamma(X),\qquad s\ge0, X\in\mathcal{M}_{2^n},
\end{equation}
see e.g.  Lemma 2.4 of
\cite{RWZ2024}.  From this, we get a useful estimate for quantum Boolean functions. 

\begin{lemma}
		If $T$ is a quantum Boolean function and $s>0$, then
	\begin{equation}\label{eq:reverse-poincare}
		\Gamma(P_sT)\le \frac1{2(e^s-1)}\one\le\frac1{2s}\one.
	\end{equation}
\end{lemma}

\begin{proof}
	For $0\le r\le s$, set
	\[
	F(r):=P_r\bigl((P_{s-r}X)^2\bigr),
	\qquad Z_r:=P_{s-r}X.
	\]
	Using
	\[
	\frac{d}{dr}P_r(B)=-P_r(\Delta B),
	\qquad
	\frac{d}{dr}P_{s-r}X=LZ_r,
	\]
	we obtain
	\begin{align*}
		F'(r)
		&=P_r\bigl(-L(Z_r^2)+(LZ_r)Z_r+Z_r(LZ_r)\bigr)\\
		&=2P_r\Gamma(Z_r).
	\end{align*}
	The second equality is precisely the definition \eqref{eq:gamma-L} applied
	to \(Z_r\).
	Thus
	\begin{equation}\label{eq:variance-interpolation}
		P_s(X^2)-(P_sX)^2
		=2\int_0^s P_r\Gamma(P_{s-r}X)\,dr.
	\end{equation}
	Moreover, \eqref{eq:gradient-input}, applied for time $r$ to $Z_r$, yields
	\begin{align*}
		\Gamma(P_sX)
		&=\Gamma(P_rZ_r)
		\le e^{-r}P_r\Gamma(Z_r),\\
		P_s(X^2)-(P_sX)^2
		&=2\int_0^sP_r\Gamma(Z_r)\,dr\\
		&\ge2\int_0^se^r\Gamma(P_sX)\,dr
		=2(e^s-1)\Gamma(P_sX).
	\end{align*}
	Thus the definition \eqref{eq:gamma-L} gives the interpolation identity
	\eqref{eq:variance-interpolation} directly, while the operator comparison
	between its integrand and \(\Gamma(P_sX)\) is exactly where the gradient
	estimate \eqref{eq:gradient-input} is used.
	For $X=T$, $T^2=\one$, $P_s\one=\one$, and $(P_sT)^2\ge0$; hence
	\[
	2(e^s-1)\Gamma(P_sT)
	\le\one-(P_sT)^2
	\le\one,
	\qquad
	\Gamma(P_sT)
	\le\frac1{2(e^s-1)}\one
	\le\frac1{2s}\one.
	\]
\end{proof}

\section{Upper commutator estimates}\label{sec 3}
For any $X\in \mathcal{M}_{2^{n}}$, we define $\beta(X)\coloneqq \max_{j\in [n]}\norm{D_jX}_2^2$ for simplicity. The main result of this section is the following  $L_{2}$-upper bound for the commutator $[X,Y]$. 

\begin{theorem}\label{main upper bound}
	Let $X,Y\in\mathcal{M}_{2^n}$ be self-adjoint elements such that
	\[
	\norm{X}_{2}\leq1,\qquad
	\deg X\leq K,\qquad \deg Y\leq d,\qquad
	1\leq d\leq K,
	\]
	where $K,d$ are integers. Then the following holds
	\begin{equation}\label{upper bound}
		\norm{[X,Y]}_{2}^{2}
		\leq 6  C_{K,d}\sqrt{\beta(X)}
		\left(\mathrm{Inf}(X)\norm{Y}_{2}^{2}+\mathrm{Inf}(Y)\right),
	\end{equation}
	where $C_{K,d}=\left(\frac{16K}{d}\right)^{d}$.
\end{theorem}

Theorem \ref{main upper bound} follows from the  following key estimate.

\begin{proposition}\label{prop:homogeneous}
	If $X,Y$ are self-adjoint with $\deg X\leq K$, $\deg Y\leq d$, and
	$1\leq d\leq K$, then
	\begin{equation}\label{eq:homogeneous}
		\norm{[X,Y]}_{2}^{2}\leq4C_{K,d}\sqrt{\beta(X)}
		\left(
		\norm{Y}_{2}\sqrt{\mathrm{Inf}(X)\mathrm{Inf}(Y)}+\norm{X}_{2}\mathrm{Inf}(Y)
		\right).
	\end{equation}
\end{proposition}

\begin{proof}[Proof of Theorem \ref{main upper bound}]
	Assume \eqref{eq:homogeneous} holds.  By $\norm{X}_{2}\leq1$ and the AM-GM inequality, we have
	\begin{equation}\label{strong}
		\begin{split}
			\norm{[X,Y]}_{2}^{2}&\leq4C_{K,d}\sqrt{\beta(X)}\left(\frac{1}{2}\|Y\|^{2}_{2}\mathrm{Inf}(X)+\frac{1}{2} \mathrm{Inf}(Y)+\mathrm{Inf}(Y)\right)\\
			&=C_{K,d}\sqrt{\beta(X)}\left(2\|Y\|^{2}_{2}\mathrm{Inf}(X)+6\mathrm{Inf}(Y)\right)\\
			&\leq6C_{K,d}\sqrt{\beta(X)}
			\left(\mathrm{Inf}(X)\norm{Y}_{2}^{2}+\mathrm{Inf}(Y)\right),
		\end{split}
	\end{equation}
	which is the desired upper bound \eqref{upper bound}. 
\end{proof}

Hence the rest of this section is devoted to proving Proposition \ref{prop:homogeneous}. To  this end, we need the following (optimal) quantum hypercontractivity and its consequence, which was established by Montanaro and Osborne \cite[Theorem~46]{MO2010}.
\begin{theorem}[Montanaro--Osborne]\label{quantum hyper}
	For every $X\in \mathcal{M}_{2^n}$ and $2<p<\infty$, the following inequality holds
	\begin{equation*}
		\norm{P_{t}X}_{p}\leq\norm{X}_{2}\Longleftrightarrow t\geq\frac12\log(p-1).
	\end{equation*}
\end{theorem}
As a consequence of Theorem \ref{quantum hyper}, the following comparison of moments theorem holds (see \cite[Corollary~51]{MO2010}), which will play an essential role of deriving the upper bound \eqref{strong}.

\begin{lemma}[Montanaro--Osborne]\label{lem:low-degree}
	For $X\in \mathcal{M}_{2^n}$ with $\deg X\leq m$, then for every $2\leq p<\infty$ we have
	\begin{equation}\label{eq:low-degree}
		\norm{X}_{p}\leq(p-1)^{m/2}\norm{X}_{2}.
	\end{equation}
\end{lemma}
\noindent Combining the noncommutative H\"older inequality with Lemma \ref{eq:low-degree}, we get the following upper bound for product.

\begin{lemma}\label{lem:ordinary-product}
	If $X$,$Y\in \mathcal{M}_{2^n}$ are self-adjoint with $\deg X\leq K$, $\deg Y\leq d$, and $1\leq d\leq K$, then 
	\begin{equation*}
		\norm{XY}_{2}\leq C_{K,d}^{1/2}\norm{X}_{2}\norm{Y}_{2}.
	\end{equation*}
\end{lemma}
\begin{proof}
	Choose
	\begin{equation*}
		p=2+\frac{d}{K},\qquad q=2+\frac{4K}{d}.
	\end{equation*}
	It is clear that $p,~q\geq 2$, and
	\[
	\frac{1}{p}+\frac{1}{q}=\frac{K}{2K+d}+\frac{d}{2d+4K}=\frac12.
	\]
	Then, by the noncommutative H\"older inequality  and Lemma \ref{lem:low-degree}, we obtain
\[
\norm{XY}_{2}\leq\|X\|_p\|Y\|_q\leq(p-1)^{K/2}(q-1)^{d/2}\norm{X}_{2}\norm{Y}_{2}.
\]
	Note that
	\begin{align*}
		(p-1)^{K}(q-1)^{d}
		&=\left(1+\frac{d}{K}\right)^{K}
		\left(1+\frac{4K}{d}\right)^{d}\\
		&\leq e^{d}\left(\frac{5K}{d}\right)^{d}
		\leq\left(\frac{16K}{d}\right)^{d}=C_{K,d}.
	\end{align*}
	where we used $(1+u)^{1/u}\leq e$ in the first inequality. The desired assertion follows. 
\end{proof}

To provide an upper bound for commutators, we introduce the family of mappings to interpolate $AB$ and $BA$ for $A$, $B\in \mathcal{M}_{2^n}$. Note first that
\[
(\sigma^{1})^{\mathsf{T}}=\sigma^{1},\quad(\sigma^{2})^{\mathsf{T}}=-\sigma^{2},\quad (\sigma^{3})^{\mathsf{T}}=\sigma^{3},
\]
where $A^{\mathsf{T}}$ stands for the transport of the matrix $A$. By this observation, for $S\subseteq[n]$, define the \emph{partial transpose} $\Theta_{S}$ as follows

\begin{equation*}
	\Theta_{S}(\sigma_{\mathbf{s}})
	=(-1)^{|\{j\in S:s_{j}=2\}|}\sigma_{\mathbf{s}},\qquad \forall~\mathbf{s}\in \{0,1,2,3\}^{n}.
\end{equation*}
It is clear that $\Theta_{S}^{2}=\id$ and $\Theta_{S}$ preserves self-adjointness and the degree on $L_{2}(\mathcal{M}_{2^{n}})$. By the Parseval identity, it follows that
\begin{equation}\label{eq:partial-transpose-isometry}
	\norm{\Theta_{S}(A)}_{2}=\norm{A}_{2}
	\quad\text{for every }A\in\mathcal{M}_{2^n}.
\end{equation}

With the help of partial transport mappings, we now interpolate $AB$ and $BA$ in the following manner. For each $S\subseteq [n]$, define a linear product (with respect to $S$) by
\begin{equation*}
	A\star_{S}B=\Theta_{S}\!\left(\Theta_{S}(A)\Theta_{S}(B)\right).
\end{equation*}
For each \[A=\bigotimes_{j=1}^nA_j,\qquad
B=\bigotimes_{j=1}^nB_j,\]
where $\{A_{j}\}_{j=1}^{n}$, $\{B_{j}\}_{j=1}^{n}\subseteq M_2(\mathbb{C})$, it is not hard to check that
\begin{equation}\label{eq:star-tensors}
	A\star_{S}B=\bigotimes_{j=1}^{n}C_{j},\qquad C_{j}=
	\begin{cases}B_{j}A_{j},&j\in S,\\
		A_{j}B_{j},&j\notin S.
	\end{cases}
\end{equation}
Thus this product reverses the local multiplication order precisely on
$S$. Particularly, we have
\begin{equation*}
	A\star_{\emptyset}B=AB,\qquad
	A\star_{[n]}B=BA
	\quad\text{for all }A,B\in\mathcal{M}_{2^n}.
\end{equation*}

We here collect a basic estimate. 
\begin{lemma}\label{star convolution}
	Let $X$,$Y\in \mathcal{M}_{2^n}$ be self-adjoint elements with $\deg X\leq K$, $\deg Y\leq d$, and $1\leq d\leq K$, then the following holds for every $S\subseteq [n]$,
	\begin{equation*}
		\norm{X\star_{S}Y}_{2}\leq C_{K,d}^{1/2}\norm{X}_{2}\norm{Y}_{2}.
	\end{equation*}
\end{lemma}
\begin{proof}
	Note that	$\Theta_{S}(X)$ preserves 
	the degree of the given $X\in \mathcal{M}_{2^n}$. 
	It follows from \eqref{eq:partial-transpose-isometry} and Lemma \ref{lem:ordinary-product} that 
	\[
	\norm{X\star_{S}Y}_{2}=\norm{\Theta_{S}(X)\Theta_{S}(Y)}_{2}\leq C_{K,d}^{1/2}\norm{\Theta_{S}(X)}_{2}\norm{\Theta_{S}(Y)}_{2}=C_{K,d}^{1/2}\norm{X}_{2}\norm{Y}_{2}.
	\]
\end{proof}

For each $j\in [n]$, we define ($X\star_{[0]}Y\coloneqq X\star_{\emptyset}Y$ for convenience)
\begin{equation*}
	\mathcal{B}_j(X,Y):=X\star_{[j-1]}Y-X\star_{[j]}Y, \qquad 1\leq j\leq n,
\end{equation*}

\begin{lemma}\label{lem:localization}
	Let $X$,$Y\in \mathcal{M}_{2^n}$ be self-adjoint elements with $\deg X\leq K$, $\deg Y\leq d$, and $1\leq d\leq K$. Then 
	\begin{enumerate}[{\rm (i)}]
		\item  We have $[X,Y]=\sum_{j=1}^{n}\mathcal{B}_j(X,Y),\label{eq:telescoping}$.
		\item For each $j\in [n]$, $\mathcal{B}_j(X,Y)=(D_{j}X)\star_{[j-1]}(D_{j}Y)-(D_{j}X)\star_{[j]}(D_{j}Y)$.
		\item For each $j\in [n]$, $E_{j}(\mathcal{B}_j(X,Y))=0$.
		\item For each $j\in [n]$,
		$\norm{\mathcal{B}_j(X,Y)}_{2}\leq2 C_{K,d}^{1/2}
		\norm{D_{j}X}_{2}\norm{D_{j}Y}_{2}$.
	\end{enumerate}
\end{lemma}
\begin{proof}
	It suffices to verify (i) for $X=\bigotimes_{k=1}^{n}A_{k}$ and $Y=\bigotimes_{k=1}^{n}B_{k}$. For each $j\in [n]$, it follows from \eqref{eq:star-tensors} that 
	\begin{equation}\label{eq:elementary-difference}
		\mathcal{B}_j(X,Y)=X\star_{[j-1]}Y-X\star_{[j]}Y=\left(\bigotimes_{k<j}B_{k}A_{k}\right)
		\otimes[A_{j},B_{j}]
		\otimes\left(\bigotimes_{k>j}A_{k}B_{k}\right).
	\end{equation}
	Summing \eqref{eq:elementary-difference} over $j\in [n]$ yields \eqref{eq:telescoping}. 
	
	Note that $X=E_{j}X+D_{j}X$ and $Y=E_{j}Y+D_{j}Y$; the partial trace is defined in \eqref{eq:Ej}. 
	By \eqref{eq:elementary-difference}, we know that 
	\[\mathcal{B}_j(E_jX,Y)=0,\quad \mathcal{B}_j(X,E_jY)=0.\]
	Then (ii) follows.
	
	Since $E_{j}$ is a partial trace, it follows from  \eqref{eq:elementary-difference} that 
	\[
	E_{j}(\mathcal{B}_j(X,Y))=\frac{1}{2}\mathrm{Tr}([A_{j},B_{j}])\left(\bigotimes_{k<j}B_{k}A_{k}\right)\otimes\mathbf{1}_{2}\otimes\left(\bigotimes_{k>j}A_{k}B_{k}\right)=0,
	\]
	which verifies (iii). 
	
	We note here that $D_{j}X$ and $D_{j}Y$ satisfy $\deg D_{j}X\leq K$, $\deg D_{j}Y\leq d$. Applying Lemma \ref{star convolution}, we have
	\[
	\|\mathcal{B}_j(X,Y)\|_{2}\leq \left\|(D_{j}X)\star_{[j-1]}(D_{j}Y)\right\|_{2}+\left\|(D_{j}X)\star_{[j]}(D_{j}Y)\right\|_{2}\leq 2 C_{K,d}^{1/2}\norm{D_{j}X}_{2}\norm{D_{j}Y}_{2}.
	\]
	This proves (iv) and completes the proof.
\end{proof}
In order to estimate the noncommutative $L_{2}$-norm from above, we need the following decomposition lemma.

\begin{lemma}\label{lem:pairing}
	For self-adjoint $X$, $Y\in \mathcal{M}_{2^n}$, we have
	\begin{equation*}
		\norm{[X,Y]}_{2}^{2}
		=2\Rea\sum_{j=1}^{n}\ip{D_{j}(XY)}{\mathcal{B}_j(X,Y)}.
	\end{equation*}
	Consequently, the Cauchy-Schwarz inequality entails
	\begin{equation*}
		\norm{[X,Y]}_{2}^{2}\leq2\sum_{j=1}^{n}\norm{D_{j}(XY)}_{2}\norm{\mathcal{B}_j(X,Y)}_{2}.
	\end{equation*}
\end{lemma}
\begin{proof}
	By the self-adjointness, we have $\|XY\|_{2}=\|YX\|_{2}$. Hence, the parallelogram law of $L_{2}(\mathcal{M}_{2^n})$-norm yields
	\[
	\norm{[X,Y]}_{2}^{2}=\norm{XY}_{2}^{2}+\norm{YX}_{2}^{2}-2\Rea\ip{XY}{YX}=2\Rea\ip{XY}{XY-YX}=2\Rea\ip{XY}{[X,Y]}.
	\]
	By Lemma \ref{lem:localization} (iii), $D_{j}(\mathcal{B}_j(X,Y))=\mathcal{B}_j(X,Y)-E_j\mathcal{B}_j(X,Y)=\mathcal{B}_j(X,Y)$. By Lemma \ref{lem:localization} (i), we have
	\[
	\ip{XY}{[X,Y]}=\sum_{j=1}^{n}\ip{XY}{\mathcal{B}_j(X,Y)}=\sum_{j=1}^{n}\ip{XY}{D_{j}\mathcal{B}_j(X,Y)}=\sum_{j=1}^{n}\ip{D_{j}(XY)}{\mathcal{B}_j(X,Y)},
	\]
	where the last equality is due the fact that $D_j=D_j^*$ for each $j$. 
\end{proof}

To provide an $L_{2}(\mathcal{M}_{2^n})$ upper bound for the commutator $[X,Y]$, by Lemma \ref{lem:pairing}, it suffices to bound the term $\|D_{j}(XY)\|_{2}$ appropriately.
\begin{lemma}\label{lem:product-derivative}
	Let $X$,$Y\in \mathcal{M}_{2^n}$ be self-adjoint elements with $\deg X\leq K$, $\deg Y\leq d$, and $1\leq d\leq K$. Then for each $j\in [n]$ we have
	\begin{equation*}
		\norm{D_{j}(XY)}_{2}
		\leq C_{K,d}^{1/2}
		\left(\norm{D_{j}X}_{2}\norm{Y}_{2}
		+\norm{X}_{2}\norm{D_{j}Y}_{2}\right).
	\end{equation*}
\end{lemma}
\begin{proof}
	Fix $j\in [n]$. Note here that
	\[
	X=E_{j}X+D_{j}X\quad\mbox{and}\quad Y=E_{j}Y+D_{j}Y,
	\]
	then the following holds
	\[
	XY=(D_{j}X)Y+(E_{j}X)Y=(D_{j}X)Y+(E_{j}X)(D_{j}Y)+(E_{j}X)(E_{j}Y).
	\]
	Also note that 
	\[D_j((E_{j}X)(E_{j}Y))=(E_{j}X)(E_{j}Y)-E_j((E_{j}X)(E_{j}Y))=0.\]
	Therefore,
	\[
	D_{j}(XY)=D_{j}\left((D_{j}X)Y+(E_{j}X)(D_{j}Y)\right).
	\]
	Since $D_{j}$ is an orthogonal projection on $L_{2}(\mathcal{M}_{2^n})$, it follows that
	\[
	\|D_{j}(XY)\|_{L_{2}}\leq \left\|(D_{j}X)Y+(E_{j}X)(D_{j}Y)\right\|_{2}\leq \|(D_{j}X)Y\|_{2}+\left\|(E_{j}X)(D_{j}Y)\right\|_{2}.
	\]
	Since
	\[
	\deg(D_{j}X)\leq K, \quad\deg(E_{j}X)\leq K,
	\]
	and
	\[
	\deg(D_{j}Y)\leq d, \quad \deg(E_{j}Y)\leq d,
	\]
	it follows from Lemma \ref{lem:ordinary-product} that
	\begin{align*}
		\|D_{j}(XY)\|_{2}&\leq \|(D_{j}X)Y\|_{2}+\|(E_{j}X)(D_{j}Y)\|_{2}\\
		&\leq C_{K,d}^{1/2}\left(\|D_{j}X\|_{2}\|Y\|_{2}+\|E_{j}X\|_{2}\|D_{j}Y\|_{2}\right)\\
		&\leq C_{K,d}^{1/2}\left(\norm{D_{j}X}_{2}\|Y\|_{2}
		+\|X\|_{2}\norm{D_{j}Y}_{2}\right).
	\end{align*}
	This completes our proof.
\end{proof}

Now we are ready to prove Proposition~\ref{prop:homogeneous}. 
\begin{proof}[Proof of Proposition~\ref{prop:homogeneous}]
	Combining Lemma \ref{lem:pairing}, Lemma \ref{lem:localization}(iv), and Lemma \ref{lem:product-derivative}, we have
	
	\begin{equation}\label{eq:before-summation}
		\begin{split}
			\norm{[X,Y]}_{2}^{2}
			&\leq4C_{K,d}\sum_{j=1}^{n}
			\norm{D_{j}X}_{2}\norm{D_{j}Y}_{2}
			\left(\norm{D_{j}X}_{2}\norm{Y}_{2}
			+\norm{X}_{2}\norm{D_{j}Y}_{2}\right)\\
			&=4C_{K,d}\left(
			\norm{Y}_{2}\sum_{j=1}^{n}\norm{D_{j}X}_{2}^{2}\norm{D_{j}Y}_{2}
			+\norm{X}_{2}\sum_{j=1}^{n}\norm{D_{j}X}_{2}\norm{D_{j}Y}_{2}^{2}
			\right).
		\end{split}
	\end{equation}
	For the first sum of \eqref{eq:before-summation}, by
	$0\leq\norm{D_{j}X}_{2}^{2}\leq\beta(X)$ and the
	Cauchy--Schwarz inequality, we obtain
	\begin{equation}\label{eq:first-influence-sum}
		\begin{split}
			\sum_{j=1}^{n}\norm{D_{j}X}_{2}^{2}\norm{D_{j}Y}_{2}
			&\leq\beta(X)^{1/2}\sum_{j=1}^{n}
			\norm{D_{j}X}_{2}\norm{D_{j}Y}_{2}\\
			&\leq\beta(X)^{1/2}
			\left(\sum_{j=1}^{n}\norm{D_{j}X}_{2}^{2}\right)^{1/2}
			\left(\sum_{j=1}^{n}\norm{D_{j}Y}_{2}^{2}\right)^{1/2}\\
			&=\left(\beta(X)\mathrm{Inf}(X)\mathrm{Inf}(Y)\right)^{1/2}.
		\end{split}
	\end{equation}
	For the second sum, it is easy to see that
	\begin{equation}\label{eq:second-influence-sum}
		\sum_{j=1}^{n}\norm{D_{j}X}_{2}\norm{D_{j}Y}_{2}^{2}
		\leq\beta(X)^{1/2}\sum_{j=1}^{n}\norm{D_{j}Y}_{2}^{2}
		=\beta(X)^{1/2}\mathrm{Inf}(Y).
	\end{equation}
	Substituting \eqref{eq:first-influence-sum} and \eqref{eq:second-influence-sum} into \eqref{eq:before-summation} yields
	\[
	\norm{[X,Y]}_{2}^{2}\leq4C_{K,d} \beta(X)^{1/2}\left(\norm{Y}_{2}\sqrt{\mathrm{Inf}(X)\mathrm{Inf}(Y)}+\norm{X}_{2}\mathrm{Inf}(Y)\right),
	\]
	which completes our proof.
\end{proof}

\section{Lower commutator estimates}\label{sec 4}

Throughout this section, we always assume that $T\in\mathcal{M}_{2^n}$ is a nontrivial quantum Boolean function.
For $t>0$ and $m\ge1$, put
\begin{equation*}
  f_t(m)=e^{-tm}-e^{-2tm},
\qquad
W_t(T):=\sum_{\mathbf{s}\ne\mathbf{0}}
f_t(|\mathbf{s}|)\widehat{T}(\mathbf{s})^2.
\end{equation*}
 Define
\[
  \Delta^{-1}\sigma_{\mathbf{s}}\coloneqq\begin{cases}
       0,&\mathbf{s}=\mathbf0,\\
       |\mathbf{s}|^{-1}\sigma_{\mathbf{s}},
         &\mathbf{s}\ne\mathbf0.
     \end{cases}
\]
Then
\begin{equation*}
  \Phi_t(T):=\Delta^{-1}(P_t-P_{2t})T
  =\int_t^{2t}P_s(T-\tr(T)\one)\,ds.
\end{equation*}
Indeed, for $\mathbf{s}\ne\mathbf0$,
\[
\Delta^{-1}(P_t-P_{2t})\sigma_{\mathbf{s}}
=\frac{e^{-t|\mathbf{s}|}-e^{-2t|\mathbf{s}|}}{|\mathbf{s}|}
\sigma_{\mathbf{s}}
=\int_t^{2t}P_s\sigma_{\mathbf{s}}\,ds.
\]
 For \(j\in[n]\) and \(a\in[3]\), let
\(\varepsilon_{ja}\) be independent uniform signs such that
\[
\mathbb P(\varepsilon_{ja}=1)
=\mathbb P(\varepsilon_{ja}=-1)=\frac12.
\]
Set 
\begin{equation}\label{eq: Yd}
Y_{\varepsilon,t}(T)
=\sum_{j=1}^n\sum_{a=1}^3
\varepsilon_{ja}\widetilde{\operatorname{ad}}_{ja}(\Phi_t(T)),
\qquad
Y_{\varepsilon,t,d}(T)=\Pi_{\le d}Y_{\varepsilon,t}(T).
\end{equation}
  Since $T$ is
fixed throughout this section and the remainder of the proof, we use the
abbreviations
\[
  W_t:=W_t(T),\quad \Phi_t:=\Phi_t(T), \quad Y_{\varepsilon,t}:=Y_{\varepsilon,t}(T), \quad  Y_{\varepsilon,t,d}:=Y_{\varepsilon,t,d}(T)
\]
from this point onward.

The main objective of this section is to establish a lower bound for the commutator in Theorem \ref{thm:lower commutator}. To this end, we first introduce a series of lemmas.

\begin{lemma}\label{lem:tests}
For each $t\in(0,\infty)$, we have
\begin{align}
  \sum_{j=1}^n\sum_{a=1}^3
    \ip{\widetilde{\operatorname{ad}}_{ja}(T)}
       {\widetilde{\operatorname{ad}}_{ja}(\Phi_t)}
    &=W_t,\label{eq:test 1}\\
  \sum_{j=1}^n\sum_{a=1}^3
    \norm{\widetilde{\operatorname{ad}}_{ja}(\Phi_t)}_2^2
    &\le tW_t,\label{eq:test 2}\\
  \sum_{j=1}^n\sum_{a=1}^3
    \mathrm{Inf}\bigl(\widetilde{\operatorname{ad}}_{ja}(\Phi_t)\bigr)
    &\le W_t,\label{eq:test 3}\\
  \sum_{j=1}^n\sum_{a=1}^3
    \bigl(\widetilde{\operatorname{ad}}_{ja}(\Phi_t)\bigr)^2
    &\le\frac t2\one.\label{eq:test 4}
\end{align}
Moreover,
\begin{equation}\label{eq:comm-lower}
  \left(
    \sum_{j=1}^n\sum_{a=1}^3
    \norm{\comm{T/2}
      {\widetilde{\operatorname{ad}}_{ja}(\Phi_t)}}_2^2
  \right)^{1/2}
  \ge \frac{W_t}{\sqrt{\mathrm{Inf}(T)}}.
\end{equation}
\end{lemma}

\begin{proof}
Note that 
Lemma~\ref{lem:adja}(iv) gives
\begin{equation}\label{eq:L}
 \sum_{j=1}^n\sum_{a=1}^3
\widetilde{\operatorname{ad}}_{ja}^*
\widetilde{\operatorname{ad}}_{ja}
=\sum_{j=1}^nD_j=\Delta.
\end{equation}
Hence, using $\Delta\Phi_t=(P_t-P_{2t})T$, we obtain
\begin{align*}
  \sum_{j=1}^n\sum_{a=1}^3
  \ip{\widetilde{\operatorname{ad}}_{ja}(T)}
     {\widetilde{\operatorname{ad}}_{ja}(\Phi_t)}
  =\ip{T}{\Delta\Phi_t}
  &=\sum_{\mathbf{s}\ne\mathbf{0}}
    \bigl(e^{-t|\mathbf{s}|}-e^{-2t|\mathbf{s}|}\bigr)
    \widehat{T}(\mathbf{s})^2
   =W_t.
\end{align*}
This is \eqref{eq:test 1}.

By \eqref{eq:L}, we also have 
$$  \sum_{j=1}^n\sum_{a=1}^3
\norm{\widetilde{\operatorname{ad}}_{ja}(\Phi_t)}_2^2
=\ip{\Phi_t}{\Delta\Phi_t}
=\sum_{\mathbf{s}\ne\mathbf{0}}
\frac{f_t(|\mathbf{s}|)^2}{|\mathbf{s}|}
\widehat{T}(\mathbf{s})^2.$$
Note that $m\ge1$,
\[
  0<f_t(m)=e^{-tm}(1-e^{-tm})
  \le1-e^{-tm}\le tm,
  \qquad
  \frac{f_t(m)^2}{m}\le t f_t(m).
\]
 Then the assertion \eqref{eq:test 2} follows. 
 
  Combining \eqref{eq:total-inf}, Lemma~\ref{lem:adja}(iii) and \eqref{eq:L}, we have
\begin{align*}
  \sum_{j=1}^n\sum_{a=1}^3
    \mathrm{Inf}\bigl(\widetilde{\operatorname{ad}}_{ja}(\Phi_t)\bigr)
  &=\sum_{j=1}^n\sum_{a=1}^3
    \ip{\widetilde{\operatorname{ad}}_{ja}(\Phi_t)}
       {\Delta\widetilde{\operatorname{ad}}_{ja}(\Phi_t)}\\
  &=\ip{\Phi_t}{\Delta^2\Phi_t}
  =\sum_{\mathbf{s}\ne\mathbf{0}}
    f_t(|\mathbf{s}|)^2\widehat{T}(\mathbf{s})^2
  \le\sum_{\mathbf{s}\ne\mathbf{0}}
    f_t(|\mathbf{s}|)\widehat{T}(\mathbf{s})^2
  =W_t,
\end{align*}
which proves \eqref{eq:test 3}.

Now we turn to prove \eqref{eq:test 4}. For each $j\in [n]$, $a\in [3]$ and $t>0$, we have 
\[\widetilde{\operatorname{ad}}_{ja}(\Phi_t)=\int_t^{2t} \widetilde{\operatorname{ad}}_{ja}(P_sT) ds.\]
Then, by  \eqref{eq:gamma-ad-square}, we have 
\[\sum_{j=1}^n\sum_{a=1}^3
\bigl(\widetilde{\operatorname{ad}}_{ja}(\Phi_t)\bigr)^2\leq t\int_t^{2t} \sum_{j=1}^n\sum_{a=1}^3
\bigl(\widetilde{\operatorname{ad}}_{ja}(P_sT)\bigr)^2 ds=t\int_t^{2t}\Gamma(P_sT)ds.\]
Thus,
\begin{align*}
  \left\|
      \sum_{j=1}^n\sum_{a=1}^3
      \bigl(\widetilde{\operatorname{ad}}_{ja}(\Phi_t)\bigr)^2
  \right\|_\infty
  &\le t\int_t^{2t}\norm{\Gamma(P_sT)}_\infty\,ds\\
  &\le t \int_t^{2t}(2s)^{-1}\,ds
  \le \frac{t}{2},
\end{align*}
where we used \eqref{eq:reverse-poincare}. 
Consequently,
\[
  0\le
  \sum_{j=1}^n\sum_{a=1}^3
    \bigl(\widetilde{\operatorname{ad}}_{ja}(\Phi_t)\bigr)^2
  \le\frac t2\one.
\]

Finally,  let us define $Q_T:L_2(\mathcal{M}_{2^n})\to L_2(\mathcal{M}_{2^n})$ by setting
\[
  Q_T(Z)=\frac{Z-TZT}{2}.
\]
Since $T=T^*$ and $T^2=\one$, it follows that 
\[
  Q_T^*=Q_T,\qquad Q_T^2=Q_T.
\]
Moreover, Lemma \ref{lem:adS} gives
\begin{align*}
  Q_T\bigl(\widetilde{\operatorname{ad}}_{ja}(T)\bigr)
  =\frac{
    \widetilde{\operatorname{ad}}_{ja}(T)
    -T\widetilde{\operatorname{ad}}_{ja}(T)T}{2}
   =\widetilde{\operatorname{ad}}_{ja}(T),
\end{align*}
and 
\begin{align*}
TQ_T(Z)
=\frac{TZ-ZT}{2}=[T/2,Z]. 
\end{align*} 
  Since $\|\cdot\|_{L_2}$ is unitary invariant, we further have 
  $$\norm{[T/2,Z]}_2=\norm{Q_T(Z)}_2,\quad Z\in L_2(\mathcal{M}_{2^n}).$$
Using the Cauchy--Schwarz inequality, we have
\begin{align*}
  W_t
  &=\sum_{j=1}^n\sum_{a=1}^3
    \ip{Q_T(\widetilde{\operatorname{ad}}_{ja}(T))}
       {\widetilde{\operatorname{ad}}_{ja}(\Phi_t)}\\
  &=\sum_{j=1}^n\sum_{a=1}^3
    \ip{\widetilde{\operatorname{ad}}_{ja}(T)}
       {Q_T(\widetilde{\operatorname{ad}}_{ja}(\Phi_t))}\\
  &\le
  \left(\sum_{j=1}^n\sum_{a=1}^3
    \norm{\widetilde{\operatorname{ad}}_{ja}(T)}_2^2\right)^{1/2}
  \left(\sum_{j=1}^n\sum_{a=1}^3
    \norm{Q_T(\widetilde{\operatorname{ad}}_{ja}(\Phi_t))}_2^2
  \right)^{1/2}\\
  &=\sqrt{\mathrm{Inf}(T)}
  \left(
    \sum_{j=1}^n\sum_{a=1}^3
    \norm{\comm{T/2}
      {\widetilde{\operatorname{ad}}_{ja}(\Phi_t)}}_2^2
  \right)^{1/2},
\end{align*}
where we used \eqref{eq:test 1} in the first equality. 
The desired assertion \eqref{eq:comm-lower} follows. 
\end{proof}

In what follows, we always write
\[
\gamma(T):=\frac{\mathrm{Inf}(T)}{\mathrm{Var}(T)}\ge1.
\]

\begin{lemma}\label{lem:window}
	 There exist an integer $\nu\ge0$ and the
corresponding dyadic number
\[
  \kappa:=2^\nu,
  \qquad 1\le\kappa<8\gamma(T),
\]
such that, upon setting
\[
  t:=\frac1\kappa,
  \qquad
  \eta:=\frac{W_t}{\mathrm{Var}(T)}
  =\frac{W_{1/\kappa}}{\mathrm{Var}(T)},
\]
one has
\begin{equation}\label{eq:window-mass}
  \eta\ge\frac1{32}\sqrt{\frac\kappa{\gamma(T)}}.
\end{equation}
\end{lemma}

\begin{proof}
Define
\[
  J:=\left\lceil\log_2\bigl(4\gamma(T)\bigr)\right\rceil,
  \qquad
  \log_2 y:=\frac{\log y}{\log2},
\]
and, for every $\nu\in\{0,\ldots,J\}$, define
\[
  \kappa_\nu:=2^\nu,
  \qquad
  t_\nu:=\frac1{\kappa_\nu},
  \qquad
  \eta_\nu:=\frac{W_{t_\nu}}{\mathrm{Var}(T)}.
\]
Also put
\[
  w_{\mathbf{s}}
  :=\frac{\widehat T(\mathbf{s})^2}{\mathrm{Var}(T)}
  \qquad(\mathbf{s}\ne\mathbf0).
\]
Then
\begin{equation*}
  \sum_{\mathbf{s}\ne\mathbf0}w_{\mathbf{s}}=1,
  \qquad
  \sum_{\mathbf{s}\ne\mathbf0}|\mathbf{s}|w_{\mathbf{s}}=\gamma(T),
  \qquad
  \eta_\nu
  =\sum_{\mathbf{s}\ne\mathbf0}
    w_{\mathbf{s}}f_{t_\nu}(|\mathbf{s}|).
\end{equation*}

For $m\ge1$ and $0\le\nu\le J$,
\[
  f_{t_\nu}(m)
  =e^{-m/2^\nu}-e^{-m/2^{\nu-1}}.
\]
For $\nu=0$, the second term on the right is $e^{-2m}$.  Therefore,
\begin{align*}
  \sum_{\nu=0}^J\eta_\nu
  &=\sum_{\mathbf{s}\ne\mathbf0}w_{\mathbf{s}}
    \left(e^{-|\mathbf{s}|/\kappa_J}
          -e^{-2|\mathbf{s}|}\right)\\
  &\ge\sum_{\mathbf{s}\ne\mathbf0}w_{\mathbf{s}}
    \left(1-\frac{|\mathbf{s}|}{\kappa_J}-e^{-2}\right)=1-\frac{\gamma(T)}{\kappa_J}-e^{-2},
\end{align*}
where  we used $e^{-s}\ge1-s$ for $s\ge0$.
By the definition of $J$,
\begin{equation}\label{eq:window-largest-scale}
  4\gamma(T)\le\kappa_J<8\gamma(T).
\end{equation}
Consequently, 
\begin{equation}\label{eq:window-total-half}
  \sum_{\nu=0}^J\eta_\nu
  \ge\frac34-e^{-2}>\frac12.
\end{equation}

On the other hand, \eqref{eq:window-largest-scale}  gives
\begin{align*}
  \sum_{\nu=0}^J\sqrt{\frac{\kappa_\nu}{\gamma(T)}}
  &=\sqrt{\frac{\kappa_J}{\gamma(T)}}
    \sum_{r=0}^J2^{-r/2}<\frac{\sqrt8}{1-2^{-1/2}}<10.
\end{align*}
If
\[
  \eta_\nu<\frac1{32}\sqrt{\frac{\kappa_\nu}{\gamma(T)}}
  \qquad \forall 0\le\nu\le J,
\]
then 
\[
  \sum_{\nu=0}^J\eta_\nu
  <\frac1{32}\sum_{\nu=0}^J
    \sqrt{\frac{\kappa_\nu}{\gamma(T)}}
  <\frac{10}{32}<\frac12,
\]
contradicting \eqref{eq:window-total-half}.  Hence, there exists
$\nu\in\{0,\ldots,J\}$ such that
\[
  \eta_\nu\ge\frac1{32}\sqrt{\frac{\kappa_\nu}{\gamma(T)}}.
\]
Taking $\kappa=\kappa_\nu$, $t=t_\nu$, and $\eta=\eta_\nu$, and using
$1\le\kappa_\nu\le\kappa_J<8\gamma(T)$, proves
\eqref{eq:window-mass}.
\end{proof}

In what follows, according to   Lemma~\ref{lem:window}, we fix  
\begin{equation}\label{eq:window-parameters}
  t=1/\kappa,
  \qquad W_t=\mathrm{Var}(T)\eta,
  \qquad q=\gamma(T)/\kappa,
  \qquad R=\max\{1,q\}.
\end{equation}
We further set
\begin{equation}\label{eq:dK}
  d=\left\lceil\kappa(8+2\log R)\right\rceil,
  \qquad
  K=\max\left\{d,
  \left\lceil\frac{2^{15}\gamma(T)^3}{\kappa^2}\right\rceil\right\}.
\end{equation}

\begin{lemma}\label{lem:low 1}
	Let $u$ be as in \eqref{eq: def u}, and let $Y_{\varepsilon,t}$ be as in \eqref{eq: Yd}. Then 
	$$	\E_\varepsilon\norm{[u(T),Y_{\varepsilon,t}]}_2^2\geq \frac{W_t^2}{\mathrm{Inf}(T)}.$$
\end{lemma}

\begin{proof}
Since $T=T^*$ and $T^2=\one$, it follows that 	\(u(T)=T/2\). 
	Note that
	\(\E_\varepsilon(\varepsilon_{ja}\varepsilon_{kb})
	=\delta_{jk}\delta_{ab}\). 
	Then 
	\begin{equation*}
		\begin{aligned}
			\E_\varepsilon\norm{[u(T),Y_{\varepsilon,t}]}_2^2
			&=\sum_{j,k=1}^n\sum_{a,b=1}^3
			\E_\varepsilon(\varepsilon_{ja}\varepsilon_{kb})
			\ip{[T/2,\widetilde{\operatorname{ad}}_{ja}(\Phi_t)]}
			{[T/2,\widetilde{\operatorname{ad}}_{kb}(\Phi_t)]}\\
			&=\sum_{j=1}^n\sum_{a=1}^3
			\norm{[T/2,\widetilde{\operatorname{ad}}_{ja}(\Phi_t)]}_2^2
			\ge\frac{W_t^2}{\mathrm{Inf}(T)},
		\end{aligned}
	\end{equation*}
	where we used \eqref{eq:comm-lower} in the last inequality. 
\end{proof}

\begin{lemma}\label{lem:low 2}
		Let $u$ be as in \eqref{eq: def u}, and let $Y_{\varepsilon,t}$ be as in \eqref{eq: Yd}. Then 
		\[\E_\varepsilon\norm{[u(T)-u(T_K),Y_{\varepsilon,t}]}_2^2\leq \frac {W_t^2}{16\mathrm{Inf}(T)},\]
		where $T_K=\Pi_{\le K}T$.  
\end{lemma}

\begin{proof}
	Note that 
	$$\E_\varepsilon Y_{\varepsilon,t}^2
	=\sum_{j,k=1}^n\sum_{a,b=1}^3
	\E_\varepsilon(\varepsilon_{ja}\varepsilon_{kb})
	\widetilde{\operatorname{ad}}_{ja}(\Phi_t)
	\widetilde{\operatorname{ad}}_{kb}(\Phi_t)\\
	=\sum_{j=1}^n\sum_{a=1}^3
	\bigl(\widetilde{\operatorname{ad}}_{ja}(\Phi_t)\bigr)^2\leq \frac{t}{2}\one,$$
	where the last inequality is due to \eqref{eq:test 4}. 
	Set $A=u(T)-u(T_K)$.  By
	Lemma~\ref{lem:u} and \eqref{eq:parent-tail},
	\begin{align*}
	\E_\varepsilon\norm{[A,Y_{\varepsilon,t}]}_2^2
		&\le4\tr\!\left(A^2\E_\varepsilon Y_{\varepsilon,t}^2\right)\\
		&\le2t\tr(A^2)
		=2t\norm{u(T)-u(T_K)}_2^2\\
		&\le2t\norm{T-T_K}_2^2
		\le\frac{2\mathrm{Inf}(T)}{\kappa(K+1)}.
	\end{align*}
	Put $V:=\mathrm{Var}(T)$.  By \eqref{eq:window-mass}, \eqref{eq:dK},
 and \eqref{eq:window-parameters},
	\[
	\eta^2\ge\frac{\kappa}{2^{10}\gamma(T)},
	\qquad
	K+1>\frac{2^{15}\gamma(T)^3}{\kappa^2}
	\ge\frac{32\gamma(T)^2}{\kappa\eta^2},
	\qquad
\gamma(T)=\frac{	\mathrm{Inf}(T)}{V},\quad W_t=V\eta.
	\]
	Therefore
	\[
	\E_\varepsilon\norm{[A,Y_{\varepsilon,t}]}_2^2
	<\frac{2\mathrm{Inf}(T)}{\kappa}
	  \frac{\kappa\eta^2}{32\gamma(T)^2}
	=\frac{\mathrm{Inf}(T)\eta^2}{16\gamma(T)^2}
	=\frac{W_t^2}{16\mathrm{Inf}(T)}.
	\]
\end{proof}

\begin{lemma}\label{lem:low 3}
			Let $u$ be as in \eqref{eq: def u}, and let $Y_{\varepsilon,t,d}$ be as in \eqref{eq: Yd}.  Then
			\[	\E_\varepsilon
			\norm{[u(T_K),Y_{\varepsilon,t}-Y_{\varepsilon,t,d}]}_2^2\leq \frac{W_t^2}{16\mathrm{Inf}(T)}.\]
\end{lemma}

\begin{proof}
By the triangle inequality,
	$$\norm{[u(T_K),Y_{\varepsilon,t}-Y_{\varepsilon,t,d}]}_2\leq 2\|u(T_K)\|_{\infty} \norm{Y_{\varepsilon,t}-Y_{\varepsilon,t,d}}_2\leq \norm{Y_{\varepsilon,t}-Y_{\varepsilon,t,d}}_2,$$
	where the last inequality is due to Lemma~\ref{lem:u}. Note that $\mathrm{Inf}(\sigma_{\mathbf{s}})=|\mathbf{s}|$ for each $\mathbf{s}\neq 0$. 
Thus, by Lemma~\ref{lem:adja}(v) and  orthogonality, we have
	\begin{align*}
		\E_\varepsilon
		\norm{[u(T_K),Y_{\varepsilon,t}-Y_{\varepsilon,t,d}]}_2^2
		&\le\E_\varepsilon
		\norm{Y_{\varepsilon,t}-Y_{\varepsilon,t,d}}_2^2\\
		&=\sum_{j=1}^n\sum_{a=1}^3
		\norm{(\mathrm{id}-\Pi_{\le d})
			\widetilde{\operatorname{ad}}_{ja}(\Phi_t)}_2^2\\
		&=\sum_{|\mathbf{s}|>d}
		\frac{f_t(|\mathbf{s}|)^2}{|\mathbf{s}|^2}
		\widehat T(\mathbf{s})^2
		\sum_{j=1}^n\sum_{a=1}^3
		\norm{\widetilde{\operatorname{ad}}_{ja}
			(\sigma_{\mathbf{s}})}_2^2\\
		&=\sum_{|\mathbf{s}|>d}
		\frac{f_t(|\mathbf{s}|)^2}{|\mathbf{s}|}
		\widehat T(\mathbf{s})^2\\
		&\le\frac{e^{-2t(d+1)}}{(d+1)^2}
		\sum_{|\mathbf{s}|>d}
		|\mathbf{s}|\widehat T(\mathbf{s})^2\le\frac{\mathrm{Inf}(T)e^{-2t(d+1)}}{(d+1)^2}.
	\end{align*}
	Here, we also used the fact: for $m\ge d+1$,
	\[
	f_t(m)^2\le e^{-2tm}\le e^{-2t(d+1)},
	\qquad
	\frac1m\le\frac{m}{(d+1)^2}.
	\]
	
	Keep in mind that we are using the symbols \eqref{eq:window-parameters}. 	Let $D=(d+1)/\kappa>8+2\log R$ (note that $t=1/\kappa$).  This further implies 
		\[
		\frac1D<\frac18,
		\qquad
		e^{-D}<e^{-8}R^{-2}.
		\]
	Using
	\[
	t(d+1)=D,
	\qquad
	\mathrm{Inf}(T)=\gamma(T)\mathrm{Var}(T),
	\qquad
	W_t=\eta \mathrm{Var}(T),
	\qquad
	d+1=\kappa D,
	\]
	we obtain
	\begin{align*}
		\left(
		\frac{4\mathrm{Inf}(T)e^{-t(d+1)}}{(d+1)W_t}
		\right)^2=\left(
		\frac{4\gamma(T)e^{-D}}{\kappa D\eta}
		\right)^2		=\left(
		\frac{4qe^{-D}}{D\eta}
		\right)^2\leq \left(
		\frac{e^{-8}R^{-2}q}{2\eta}
		\right)^2.
	\end{align*}
	By \eqref{eq:window-mass}, we have 
$$	\frac{e^{-8}R^{-2}q}{2\eta}\leq 16 e^{-8}R^{-2}q^{3/2}<1$$
since 
\[
q^{3/2}R^{-2}
=
\begin{cases}
	q^{-1/2}\le1, & q\ge1,\\[2mm]
	q^{3/2}<1, & 0<q<1,
\end{cases}
\]
as $R=\max\{1,q\}$. Consequently, we have
$$\E_\varepsilon
\norm{[u(T_K),Y_{\varepsilon,t}-Y_{\varepsilon,t,d}]}_2^2\leq \frac{\mathrm{Inf}(T)e^{-2t(d+1)}}{(d+1)^2}\leq \frac{W_t^2}{16\mathrm{Inf}(T)}.$$
This finishes the proof of the lemma.
\end{proof}

The following lower commutator estimate is the main result of this section. 
\begin{theorem}\label{thm:lower commutator}
Let $u$ be as in \eqref{eq: def u}. 
Let $t,d,K$ be as in \eqref{eq:window-parameters} and \eqref{eq:dK}, and let $T_K=\Pi_{\le K} T$, and $Y_{\varepsilon,t,d}$ be as in \eqref{eq: Yd}. Then
\begin{equation*}
  \E_\varepsilon
    \norm{\comm{u(T_K)}{Y_{\varepsilon,t,d}}}_2^2
  >\frac{W_t^2}{4\mathrm{Inf}(T)}.
\end{equation*}
\end{theorem}
\begin{proof}
According to the definition of $[\cdot, \cdot]$, we have
\begin{align*}
  [u(T_K),Y_{\varepsilon,t,d}]
  &=[u(T),Y_{\varepsilon,t}]
    -[u(T)-u(T_K),Y_{\varepsilon,t}]
    -[u(T_K),Y_{\varepsilon,t}-Y_{\varepsilon,t,d}]. 
\end{align*}
Consequently, by the   triangle inequality and Lemmas \ref{lem:low 1}-\ref{lem:low 3}, we have
\begin{align*}
	  \left(\E_\varepsilon
	\norm{[u(T_K),Y_{\varepsilon,t,d}]}_2^2\right)^{1/2}
	&\ge\left(	\E_\varepsilon\norm{[u(T),Y_{\varepsilon,t}]}_2^2\right)^{1/2}\\
	&\quad-\left(\E_\varepsilon\norm{[u(T)-u(T_K),Y_{\varepsilon,t}]}_2^2\right)^{1/2}\\
	&\quad-\left(	\E_\varepsilon
	\norm{[u(T_K),Y_{\varepsilon,t}-Y_{\varepsilon,t,d}]}_2^2\right)^{1/2}\\
&\quad \geq \frac {W_t}{\sqrt{\mathrm{Inf}(T)}}-\frac {W_t}{4\sqrt{\mathrm{Inf}(T)}}-\frac {W_t}{4\sqrt{\mathrm{Inf}(T)}}= \frac {W_t}{2\sqrt{\mathrm{Inf}(T)}}.
\end{align*}
The desired assertion follows. 
\end{proof}

\section{Proof of Theorem \ref{thm:main}}\label{sec 5}

In this section, we provide the proof of Theorem \ref{thm:main}. 

\begin{proof}[Proof of Theorem \ref{thm:main}]
	We first prove \eqref{eq:exp-form}. Let $u$ be as in \eqref{eq: def u}. 
	Let $t,d,K$ be as in \eqref{eq:window-parameters} and \eqref{eq:dK}, and let $T_K=\Pi_{\le K} T$, and $Y_{\varepsilon,t,d}$ be as in \eqref{eq: Yd}. 
	
	By Theorem \ref{thm:lower commutator},  we have 
	\begin{equation}\label{eq:e1}
	\frac{W_t^2}{4\mathrm{Inf}(T)}<\E_\varepsilon
		\norm{\comm{u(T_K)}{Y_{\varepsilon,t,d}}}_2^2.
	\end{equation}
On the other hand, applying Lemma~\ref{lem:u} and
Theorem~\ref{main upper bound} (with
\(X=T_K\) and \(Y=Y_{\varepsilon,t,d}\) there), we have
\begin{align}\label{eq:e2}
  \E_\varepsilon\norm{\comm{u(T_K)}{Y_{\varepsilon,t,d}}}_2^2
  &\le\E_\varepsilon\norm{\comm{T_K}{Y_{\varepsilon,t,d}}}_2^2\notag\\
  &\le32\left(\frac{16K}{d}\right)^d
    \sqrt{\beta(T_K)}\,
    \E_\varepsilon\left(
      \mathrm{Inf}(T_K)\norm{Y_{\varepsilon,t,d}}_2^2
      +\mathrm{Inf}(Y_{\varepsilon,t,d})\right).
\end{align}
Observe that 
$$\beta(T_K)\leq \beta(T),\qquad \mathrm{Inf}(T_K)\leq \mathrm{Inf}(T). $$
By orthogonality, the $L_2$-contractivity of
$\Pi_{\le d}$,  \eqref{eq:test 2} and \eqref{eq:test 3}, we have 
\begin{equation*}
	\E_\varepsilon\norm{Y_{\varepsilon,t,d}}_2^2\leq \sum_{j=1}^n\sum_{a=1}^3
	\norm{\widetilde{\operatorname{ad}}_{ja}(\Phi_t)}_2^2\le tW_t,
\end{equation*}
and
\begin{equation*}
	\E_\varepsilon \mathrm{Inf}(Y_{\varepsilon,t,d})\le \sum_{j=1}^n\sum_{a=1}^3
	\mathrm{Inf}\bigl(\widetilde{\operatorname{ad}}_{ja}(\Phi_t)\bigr)\leq W_t.
\end{equation*}
Thus, it follows from \eqref{eq:e2}  that 
\begin{align}\label{eq:e3}
	\E_\varepsilon\norm{\comm{u(T_K)}{Y_{\varepsilon,t,d}}}_2^2
	\leq 32\left(\frac{16K}{d}\right)^d
	\sqrt{\beta(T)}\,
	\left(
	\mathrm{Inf}(T)t W_t
	+W_t\right).
\end{align}

Combining \eqref{eq:e1} and \eqref{eq:e3} yields
\[
  \sqrt{\beta(T)}>
  \frac{W_t}{128\mathrm{Inf}(T)\bigl(1+t\mathrm{Inf}(T)\bigr)(16K/d)^d},
\]
and thus, 
\begin{equation}\label{eq:log-master}
 \log \frac{1}{\beta(T)}\le2d\log\frac{16K}{d}
  +2\log\frac{128\mathrm{Inf}(T)\bigl(1+t\mathrm{Inf}(T)\bigr)}{W_t}.
\end{equation}

From \eqref{eq:dK}, we know that $  8\kappa \leq d\le\kappa(9+2\log R).$
Hence,
\begin{align*}
  \frac Kd
  &\le\max\left\{
    1,\frac{2^{15}\gamma(T)^3}{8\kappa \cdot \kappa^2}+\frac{1}{8\kappa}\right\}\le1+2^{12}q^3+\frac18
   \le2^{13}R^3.
\end{align*}
Here we also used symbols introduced in \eqref{eq:window-parameters}. 
If $0<q:= \gamma(T)/\kappa <1$, then $R=1$. In this case, 
$$2d\log(16K/d)\leq 18 \kappa \times 17 \log 2
\le 1700 \gamma(T) ,$$
where we used the fact $\kappa<8 \gamma(T)$ proved in Lemma \ref{lem:window}. 
If $q\geq 1$, then $R=q$. In this case,
\begin{align*}
	2d\log(16K/d)&\leq 2\kappa (9+\log q) \times (17 \log 2+3\log q)
	,\\
	&=\gamma(T)2 q^{-1}(9+\log q) \times (17 \log 2+3\log q)=: \gamma (T) g(q).
\end{align*}
Note that the function $g$ is decreasing in $[1,\infty)$. Hence, 
$$2d\log(16K/d)\leq \gamma(T)g(1)\leq 1700 \gamma(T). $$

Next, note that $W_t=\eta \mathrm{Var}(T)$,
\[
  t\mathrm{Inf}(T)=\mathrm{Var}(T)q\le q,
  \qquad
  \frac1\eta\le32\sqrt q,
  \qquad
  1\le R\le\gamma(T),
\]
and hence
\begin{align*}
  \frac{128\mathrm{Inf}(T)\bigl(1+t\mathrm{Inf}(T)\bigr)}{W_t}
  =\frac{128\gamma(T)\bigl(1+\mathrm{Var}(T)q\bigr)}\eta
  &\le2^{12}\gamma(T)\sqrt q(1+q)\\
  &\le2^{13}\gamma(T)R^{3/2},
\end{align*}
where
\[
  \sqrt q(1+q)
  \le\begin{cases}
    2q^{3/2}=2R^{3/2},&q\ge1,\\
    2=2R^{3/2},&0<q<1.
  \end{cases}
\]
Consequently,
\begin{equation*}
  2\log\frac{128\mathrm{Inf}(T)\bigl(1+t\mathrm{Inf}(T)\bigr)}{W_t}
  \le2\log\bigl(2^{13}\gamma(T)R^{3/2}\bigr)
  <20+5\log\gamma(T)
  \le25\gamma(T).
\end{equation*}

Combining \eqref{eq:log-master}, and the last two paragraphs, we conclude that 
\[
  \log\frac1{\beta(T)}
  \le1700\gamma(T)+25\gamma(T)\le2^{11}\gamma(T),
\]
which proves \eqref{eq:exp-form}.

Now we deduce the KKL inequality from \eqref{eq:exp-form}.
Let (notice that $\mathrm{Var}(T)\in (0,1]$)
\[
  z(T)=\frac{n}{\mathrm{Var}(T)}\ge \max\{n,2\}. 
\]
If
$\gamma(T)=\frac{\mathrm{Inf}(T)}{\mathrm{Var}(T)}\ge \log z(T)/2^{12}$, then gives
\begin{equation*}
  \beta(T)=\max_{j\in [n]}\|D_jT\|_2^2\ge\frac{\mathrm{Inf}(T)}n=\frac{\gamma(T)}{z(T)}
  \ge\frac{\log z(T)}{2^{12}z(T)}
  \geq \frac{\mathrm{Var}(T)}{2^{12}}\frac{\log n}{n}.
\end{equation*}
If $\gamma(T)<\log z(T)/2^{12}$, then \eqref{eq:exp-form} gives
\[
  \log(1/\beta(T))\leq \frac{2^{11}\mathrm{Inf}(T)}{\mathrm{Var}(T)}
  <\frac{1}{2}\log z(T).
\]
and hence, 
$$ 
\beta(T)>z(T)^{-1/2}=\frac{z(T)^{1/2}}{z(T)}
\geq \frac{\log z(T)}{z(T)}
\geq \frac{\log z(T)}{2^{12}z(T)}
\geq \frac{\mathrm{Var}(T)}{2^{12}}\frac{\log n}{n}.$$
Here we used the fact that  $\log t<\sqrt t$ for $t\ge2$
The proof of Theorem \ref{thm:main} is complete. 
\end{proof}

\section*{Acknowledgment}
The authors are grateful to PhD student Zhuo Cheng for helpful communications. This work was supported by the National Natural Science Foundation of China (Grant Nos. 12125109 \& W2411005 \& 12671165); the Natural Science Foundation of Hunan Province (Grant Nos: 2025ZYJ002, 2024JJ1010 \& 2024RC3040); the Scientific Research Fund of Hunan Provincial Education Department (Grant Nos. 25A0009 \& 25B0008).

During the preparation of this work, the authors used GPT-5.6 Sol to assist with language editing and refinement of mathematical arguments. The author takes full responsibility for the content of this paper.


\begin{thebibliography}{10}

\bibitem{BIM2023}
D.~Beltran, P.~Ivanisvili, and J.~Madrid, \emph{On sharp isoperimetric
  inequalities on the hypercube}, arXiv preprint arXiv:2303.06738 (2023).

\bibitem{BGX2024}
D. Blecher, L. Gao, and B. Xu, \emph{Geometric influences on quantum Boolean cubes}, J. Funct. Anal. \textbf{289} (2025), no.~11, 111132.

\bibitem{BKKKL1992}
J.~Bourgain, J.~Kahn, G.~Kalai, Y.~Katznelson, and N.~Linial, \emph{The
  influence of variables in product spaces}, Israel J. Math. \textbf{77}
  (1992), no.~1-2, 55--64.

\bibitem{CL2012}
D.~Cordero-Erausquin and M.~Ledoux, \emph{Hypercontractive measures,
  {T}alagrand's inequality, and influences}, Geometric aspects of functional
  analysis, Lecture Notes in Math., vol. 2050, Springer, Heidelberg, 2012,
  pp.~169--189.

\bibitem{EG2020}
R.~Eldan and R.~Gross, \emph{Concentration on the {B}oolean hypercube via
  pathwise stochastic analysis}, S{TOC} '20---{P}roceedings of the 52nd
  {A}nnual {ACM} {SIGACT} {S}ymposium on {T}heory of {C}omputing, ACM, New
  York, 2020, pp.~208--221.

\bibitem{EG2022}
\bysame, \emph{Concentration on the {B}oolean hypercube via pathwise stochastic
  analysis}, Invent. Math. \textbf{230} (2022), no.~3, 935--994.

\bibitem{EKLM2022}
R.~Eldan, G.~Kindler, N.~Lifshitz, and D.~Minzer, \emph{Isoperimetric
  inequalities made simpler}, arXiv preprint arXiv:2204.06686 (2022).

\bibitem{Fr1998}
E.~Friedgut, \emph{Boolean functions with low average sensitivity depend on few
  coordinates}, Combinatorica \textbf{18} (1998), no.~1, 27--35.

\bibitem{IZ2024}
P.~Ivanisvili and H.~Zhang, \emph{On the Eldan--Gross inequality}, J. Funct. Anal. \textbf{290} (2026), no.~4, Paper No.~111255.

\bibitem{JLZZ2024}
Y.~Jiao, W.~Lin, S.~Luo, and D.~Zhou, \emph{Quantum {KKL}-type inequalities
  revisited}, arXiv:2411.12399 (2024).

\bibitem{JLZ2025}
Y.~Jiao, S.~Luo, and D.~Zhou, \emph{Functional {$L_1$}-{$L_p$} inequalities in
  the {CAR} algebra}, J. Funct. Anal. \textbf{288} (2025), no.~2, Paper No.
  110700.

\bibitem{KKL1988}
J.~Kahn, G.~Kalai, and N.~Linial, \emph{The influence of variables on {B}oolean
  functions}, 29th {A}nnual {S}ymposium on {F}oundations of {C}omputer
  {S}cience, IEEE Comput. Soc. Press, Washington, DC, [1988] \copyright 1988,
  pp.~68--80.

\bibitem{KKKMS2021}
E.~Kelman, S.~Khot, G.~Kindler, D.~Minzer, and M.~Safra, \emph{Theorems of
  {KKL}, {F}riedgut, and {T}alagrand via random restrictions and
  {L}og-{S}obolev inequality}, 12th {I}nnovations in {T}heoretical {C}omputer
  {S}cience {C}onference, vol. 185, Schloss Dagstuhl. Leibniz-Zent. Inform.,
  Wadern, 2021, pp.~Art. No. 26, 17.

\bibitem{MO2010}
A.~Montanaro and T.~J. Osborne, \emph{Quantum {B}oolean functions}, Chic. J.
  Theoret. Comput. Sci. (2010).

\bibitem{OW2013}
R.~O'Donnell and K.~Wimmer, \emph{K{KL}, {K}ruskal-{K}atona, and monotone
  nets}, SIAM J. Comput. \textbf{42} (2013), no.~6, 2375--2399.

\bibitem{PS2011}
D.~Potapov and F.~Sukochev, \emph{Operator-{L}ipschitz functions in
  {S}chatten-von {N}eumann classes}, Acta Math. \textbf{207} (2011), no.~2,
  375--389.

\bibitem{Ros2020}
G.~Rosenthal, \emph{Ramon van {H}andel's remarks on the {D}iscrete cube},
  (2020).

\bibitem{RWZ2024}
C.~Rouz\'{e}, M.~Wirth, and H.~Zhang, \emph{Quantum {T}alagrand, {KKL} and
  {F}riedgut's theorems and the learnability of quantum {B}oolean functions},
  Comm. Math. Phys. \textbf{405} (2024), no.~4, Paper No. 95, 47.

\bibitem{Ta1994}
M.~Talagrand, \emph{On {R}usso's approximate zero-one law}, Ann. Probab.
  \textbf{22} (1994), no.~3, 1576--1587.

\bibitem{Ta1997}
\bysame, \emph{On boundaries and influences}, Combinatorica \textbf{17} (1997),
  no.~2, 275--285.
\end{thebibliography}

\providecommand{\bysame}{\leavevmode\hbox to3em{\hrulefill}\thinspace}
\providecommand{\MR}{\relax\ifhmode\unskip\space\fi MR }
\providecommand{\MRhref}[2]{%
  \href{http://www.ams.org/mathscinet-getitem?mr=#1}{#2}
}
\providecommand{\href}[2]{#2}

\end{document}